\documentclass[11pt,article]{llncs}

\usepackage{subfigure}
\usepackage{amsmath}
\usepackage{amssymb}
\usepackage{graphicx}
\usepackage[usenames,dvipsnames]{color}
\usepackage{hyperref}

\def\code{\mathcal{C}}
\def\stego{\mathcal{S}}
\def\ball{\mathcal{B}}

\def\synd{\mathbf{m}}
\def\ext{p}
\def\bound{T}
\def\r{\lambda}
\def\Rrho{\rho}

\newcommand{\F}{\mathbb F}
\newcommand{\ceil}[1]{\lceil #1\rceil}

\title{Ensuring message embedding in  wet paper
  steganography}

\author{Daniel Augot \inst{1} \and Morgan Barbier \inst{1} \and
  Caroline Fontaine\inst{2}}
\institute{Computer science laboratory of \'Ecole Polytechnique\\
  INRIA Saclay -- \^Ile de France \and CNRS/Lab-STICC and 
  T\'el\'ecom Bretagne, Brest, France}
\authorrunning{ Barbier \and Augot \and Fontaine}

\begin{document}
\maketitle

\sloppypar

\begin{abstract}
  Syndrome coding has been proposed by Crandall in 1998 as a method to
  stealthily embed a message in a cover-medium through the use of
  bounded decoding. In 2005, Fridrich \emph{et al.}  introduced wet
  paper codes to improve the undetectability of the embedding by
  enabling the sender to lock some components of the cover-data,
  according to the nature of the cover-medium and the
  message. Unfortunately, almost all existing methods solving the
  bounded decoding syndrome problem with or without locked components
  have a non-zero probability to fail.
  In this paper, we introduce a randomized syndrome coding, which
  guarantees the embedding success with probability one. We analyze
  the parameters of this new scheme in the case of perfect codes.
\end{abstract}


\section{Introduction}

Hiding messages in innocuous-looking \emph{cover-media} in a
\emph{stealthy} way, steganography is the art of stealth
communications. The sender and receiver may proceed by cover
selection, cover synthesis, or cover modification to exchange
messages. Here, we focus on the cover modification scenario, where the
sender chooses some \emph{cover-medium} in his library, and modifies
it to carry the message she wants to send. Once the cover-medium is
chosen, the sender extracts some of its components to construct a
\emph{cover-data} vector. Then, she modifies it to embed the
message. This modified vector, called the \emph{stego-data}, leads
back to the \emph{stego-medium} that is communicated to the recipient.
In the case of digital images, the insertion may for example consist
in modifying some of the images components, \emph{e.g.} the luminance
of the pixels or the values of some transform (DCT or wavelet)
coefficients. For a given transmitted document, only the sender and
receiver have to be able to tell if it carries an hidden message or
not~\cite{simmons_crypto_83}. This means that the \emph{stego-media},
which carry the messages, have to be statistically indistinguishable
from original media~\cite{cachin_wih_98,cachin_ic_04}. But statistical
detectability of most steganographic schemes increases with
\emph{embedding distortion}~\cite{kodovsky_mmsec_07}, which is often
measured with the number of embedding changes. Hence it is of
importance for the sender to embed the message while modifying as less
components of the cover-data as possible.

In 1998, Crandall proposed to model the embedding and extraction
process with the use of linear error correcting codes. He proposed to
use Hamming codes, which are covering
codes~\cite{Crandall}. The key idea of this approach, called
\emph{syndrome coding}, or \emph{matrix embedding}, is to modify the
cover-data to obtain a stego-data lying in the \emph{right}
\emph{coset} of the code, its \emph{syndrome} being precisely equal to
the message to hide. Later on, it has been showed that designing
steganographic schemes is precisely equivalent to designing covering
codes~\cite{bierbrauer_oncrandall_01,galand_itw_03,galand_pit_09},
meaning that this covering codes approach is not
restrictive. Moreover, it has been shown to be really helpful and
efficient to minimize the embedding
distortion~\cite{bierbrauer_oncrandall_01,galand_itw_03,galand_pit_09,BierbrauerFridrich}. It
has also been made popular due to its use in the famous steganographic
algorithm F5~\cite{F5}. For all these reasons, this approach is of
interest.

The process which states which components of the cover-data can
actually be modified is called the \emph{selection
  channel}~\cite{LimitsSteganography}.  Since the message embedding
should introduce as little distortion as possible, the selection
channel is of utmost importance.  The selection channel may be
arbitrary, but a more efficient approach is to select it dynamically
during the embedding step, accordingly to the cover-medium and the
message. This leads to a better undetectability, and makes  attacks
on the system harder to run, but in this context the extraction of the
hidden message is more difficult as the selection channel is only known
to the sender, and not to the recipient. {\em Wet Paper Codes} were
introduced to tackle this non-shared selection channel, through the
notions of \emph{dry} and \emph{wet}
components~\cite{Fridrich_writingon}. 
By analogy with a sheet of paper
that has been exposed to rain, we can still write easily on {dry}
spots whereas we cannot write on {wet} spots. The idea is,
adaptively to the message and the cover-medium, to \emph{lock} some
components of the cover-data --- the wet components --- to prevent
them being modified. The other components --- the dry components ---
of the cover-data remain free to be modified to embed the
message. 

Algorithmically speaking, syndrome coding provides the recipient an
easy way to access the message, through a simple syndrome
computation. But to embed the message, the sender has to tackle an
harder challenge, linked with bounded syndrome decoding. It has been
shown that if random codes may seem interesting for their asymptotic
behavior, their use leads to solve really hard problems: syndrome
decoding and covering radius computation, which are proved to be
NP-complete and $\Pi_2$-complete
respectively~\cite{vardy_i3eit_97,mcloughlin_i3eit_84}. Moreover, no
efficient decoding algorithm is known, for generic, or random,
codes. Hence, attention has been given on structured codes
to design Wet Paper Codes: Hamming
codes~\cite{Crandall,EfficientWetPaperCodes}, Simplex
codes~\cite{fridrich_i3eisf_06_3}, BCH
codes~\cite{schonfeld_mmsec_06,schonfeld_ih_07,zhang_ih_09,sachnev_mmsec_09,ouldmedeni_jatit_10},
Reed-Solomon codes~\cite{fontaine_ih_07,FontaineGaland}, perfect
product codes~\cite{rifa_dsp_09,rifa_isita_10}, low density generator
matrix
codes~\cite{fridrich_spie_07,zhang_ih_08,zhang_i3eit_10,filler_wifs_09},
and convolutional
codes~\cite{TreillisCode,filler_wifs_10,filler_i3eifs_11}.

Embedding techniques efficiency is usually evaluated through their
relative payload (number of message symbols per cover-data
(modifiable) symbol) and average embedding efficiency (average number
of message symbols per cover-data modification).
Today, we can find in the literature quasi-optimal codes in terms of
average embedding efficiency and
payload~\cite{fridrich_spie_07,zhang_ih_08,zhang_i3eit_10,fridrich_i3eifs_09,filler_wifs_09}.
Nevertheless, we are interested here in another criterion, which is
usually not discussed: the probability for the embedding to fail. In
fact, the only case for which it never fails is when using perfect
codes (a), without locking any component of the cover-data (b). But
very few codes are perfect (namely the Hamming and Golay codes), and
their average embedding efficiency is quite low. Moreover it is really
important in practice to be able to lock some components of the
cover-data. Hence, efficient practical schemes usually do not satisfy either
condition (a) or condition (b), leading to a non-zero probability for
the embedding to fail.
And this probability increases with the number of locked components.
More precisely, syndrome coding usually divides the whole message into
fragments, that are separately inserted in different cover-data
vectors (coming from one or several cover-medium). Inserting each
fragment involves finding a low weight solution of a linear system
which may not always have a solution for a given set of locked
components. Consequently, the probability that the whole message can
be embedded decreases exponentially with the number of fragments to
hide and with the number of locked components
\cite{EfficientWetPaperCodes}

Hence, we have to decide what to do when embedding fails.  In the
common scenario where the sender has to choose a cover-medium in a
huge collection of documents, she can drop the cover-medium that leads
to a failure and choose another one, iterating the process until
finding a cover-medium that is adequate to embed the message. Another
solution may be to cut the message into smaller pieces, in order to
have shorter messages to embed, and a lower probability of failure. If
none of these is possible, for example if the sender only has few
pieces of content, she may unlock some locked
components~\cite{TreillisCode} to make the probability of failure
decrease. But, even doing this modified embedding, and decreasing the
probability of failure, the sender will not be able to drop it to
zero, except if she falls back to perfect codes without locked
components.

In this paper, we consider the ``worst case'' scenario, where the
sender does not have too much cover documents to hide his message in,
and then absolutely needs embedding to succeed. This scenario is not
the most studied one, and concerns very constrained situations. Our
contribution is to propose an embedding scheme that will never fail,
and does not relax the management of locked components of his
cover-data to make embedding succeed.  It is, to our knowledge, the
first bounded syndrome coding scheme that manages locked components
while guaranteeing the complete embedding of the message for any code,
be it perfect or not. To do so, we modify the classical syndrome
coding approach by using some part of the syndrome for randomization.
Of course, as the message we can embed is now shorter than the
syndrome, there is a loss in terms of embedding efficiency. We analyze
this loss in the case of linear perfect codes. Moreover, inspired by
the ZZW construction~\cite{zhang_ih_08}, we show how the size of the
random part of the syndrome, which is dynamically estimated during
embedding, can be transmitted to the recipient without any additional
communication.

The paper is organized as follows. Basic definitions and notation on
both steganography and syndrome coding are introduced in
Section~\ref{sec:stegcoding}. The traditional syndrome coding approach
is recalled at the end of this section. In
Section~\ref{sec:ourscheme}, we show how to slightly relax the
constraints on the linear system to make it always solvable, and also
estimate the loss of embedding efficiency. 
We discuss the behavior of our scheme in the case of the Golay and
Hamming perfect codes in Section~\ref{sec:ex}. Finally, as our
solution uses a parameter $r$ that is dynamically computed during
embedding, we provide in Section~\ref{sec:zzw} a construction that
enables to transmit $r$ to the recipient through the stego-data
itself, that is, without any parallel or side-channel
communication. We finally conclude in Section~\ref{sec:conclusion}.

\section{Steganography and coding theory}
\label{sec:stegcoding}

\subsection{Steganographic schemes}

We define a \emph{stego-system} (or a \emph{steganographic scheme}) by
a pair of functions, $Emb$ and $Ext$. $Emb$ embeds the message
$\mathbf{m}$ in the cover-data $\mathbf{x}$, producing the stego-data
$\mathbf{y}$, while $Ext$ extracts the message $\mathbf{m}$ from the
stego-data $\mathbf y$. To make the embedding and extraction work
properly, these functions have to satisfy the following properties.

\begin{definition}[Stego-System]
  Let $\mathcal{A}$ a finite alphabet, $r, n \in \mathbb{N}$ such that $r<n$,
  $\mathbf{x} \in \mathcal{A}^n$ denote the {\em cover-data}, $\synd
  \in \mathcal{A}^r$ denote the  message to
    embed, and $\bound$ be a strictly positive integer. A {\em
      stego-system} is defined by a pair of functions $Ext$ and $Emb$
    such that:
    \begin{eqnarray}
      Ext(Emb(\mathbf{x},\synd)) & =&  \synd \label{req:emb}\\
      d(\mathbf{x},Emb(\mathbf{x},\synd)) & \le & \bound \label{req:dist}
    \end{eqnarray}
where $d(.,.)$ denoting the Hamming distance over $\mathcal{A}^n$.
\end{definition}
Two quantities are usually used to compare stego-systems:
the embedding efficiency and the relative payload, which are defined
as follows.

\begin{definition}[Embedding efficiency]
  The {\em average embedding efficiency} of a stego-system, is usually
  defined by the ratio of the number of message symbols we can embed
  by the average number of symbols changed. We denote it by $e$.
\end{definition}

\begin{definition}[Relative payload]
  The {\em relative payload} of a stego-system, denoted by $\alpha$,
  is the ratio of the number of message symbols we can embed by the
  number of (modifiable) symbols of covered data.
\end{definition}
For $q$-ary syndrome coding, the sphere-covering bound gives an upper
bound for the embedding efficiency \cite{fridrich_i3eifs_09}. Note
that it is usually stated for binary case, using the binary entropy
function.

\begin{proposition}[Sphere-covering bound]\label{prop:embeffbound}
    For any $q$-ary stego-system $\stego$, the {\em sphere-covering bound}
    gives 
    $$
    e \le \frac{\alpha}{\mathcal{H}_q^{-1}(\alpha)},
    $$ where $\mathcal{H}_q^{-1}()$ denotes the inverse function of
    the $q$-ary entropy \mbox{$\mathcal{H}_q(x) = x \log_q(q-1) -x \log_q(x) -
      (1-x)\log_q(1-x)$} on $[0,1-1/q]$, and $\alpha$
    is the relative payload associated with $\stego$.
\end{proposition}

\subsection{From coding theory to steganography}
\label{sec:crandall}

This section recalls how coding theory may help embedding the message,
and how it tackles the non-shared selection channel paradigm. In the
rest of paper, the finite alphabet $\mathcal{A}$ is a finite field of
cardinal $q$, denoted $\mathbb{F}_q$.

Here we focus on the use of linear codes, which is the most
studied. Let $\code$ be a $[n,k,d]_q$-linear code, with parity check
matrix $H$ and covering radius $\Rrho$ --- it is the smallest
integer such that the balls of radius $\Rrho$ centered on $\code$'s
codewords cover the whole ambient space $\mathbb{F}_q^n$.  A syndrome
coding scheme based on $\code$ basically modify the cover-data
$\mathbf{x}$ in such a way that the syndrome $\mathbf{y}H^t$ of the
stego-data $\mathbf{y}$ will precisely be equal to the message
$\mathbf{m}$. Determining which symbols of $\mathbf{x}$ to modify
leads to finding a solution of a particular linear system that
involves the parity check matrix $H$.  This embedding approach has been
introduced by Crandall in 1998 \cite{Crandall}, and is called
\emph{syndrome coding} or \emph{matrix embedding}.

We formulate several embedding problems. The first one addresses
only Eq.~(\ref{req:emb}) requirements, whereas the second one also
tackles Eq.~(\ref{req:dist}).

\begin{problem}[Syndrome coding problem]\label{Prob:SCP}
  Let $\code$ be an $[n,k,d]_q$ linear code, $H$ be a parity check
  matrix of $\code$, $\mathbf{x} \in \mathbb{F}_q^n$ be a cover-data,
  and $\synd 
  \in \mathbb{F}_q^{n-k}$ be the message to be hidden in $\mathbf{x}$.
  The \emph{syndrome coding problem} consists in finding
  $\mathbf{y}\in \mathbb{F}_q^n$ such that $\mathbf{y}H^t = \synd$.
\end{problem}

\begin{problem}[Bounded syndrome coding problem]\label{Prob:BSCP}
  Let $\code$ be an $[n,k,d]_q$ linear code, $H$ be a parity check
  matrix of $\code$, $\mathbf{x} \in \mathbb{F}_q^n$ be a cover-data, $\synd
  \in \mathbb{F}_q^{n-k}$ be the message to be hidden in $\mathbf{x}$, and $\bound
  \in \mathbb{N}^*$ be an upper bound on the number of authorized
  modifications.  The \emph{bounded syndrome coding problem} consists in
  finding $\mathbf{y}\in \mathbb{F}_q^n$ such that $\mathbf{y}H^t = \synd$,
    and $d(\mathbf{x},\mathbf{y}) \le \bound$.
\end{problem}
Let us first focus on Problem~\ref{Prob:SCP}, which leads to
describing the stego-system in terms of syndrome computation:
\begin{align*}
\mathbf{y}& = Emb(\mathbf{x}, \synd) =  \mathbf{x} + D(\synd -
\mathbf{x}H^t),\\
Ext(\mathbf{y}) & =  \mathbf{y} H^t,
\end{align*}
where $D$ is the mapping associating to a syndrome $\synd$, a vector
whose syndrome is precisely equal to $\synd$. The mapping $D$ is thus
directly linked to a decoding function $f_\code$ of $\code$ of
arbitrary radius $\bound_f$, defined as \mbox{$f_\code : \mathbb{F}_q^n
\longrightarrow \code \cup \lbrace?\rbrace, $} such that for all
$\mathbf{y} \in \mathbb{F}_q^n$, either $f_\code(\mathbf{y}) = ?$, or
$d(\mathbf{y},f_\code(\mathbf{y})) \le \bound_f$.

The Hamming distance between vectors $\mathbf{x}$ and $\mathbf{y}$ is
then less than or equal to $\bound_f$.  Since decoding general codes
is NP-Hard \cite{DecodingNPHard}, finding such a mapping $D$ is not
tractable if $\code$ does not belong to a family of codes we can
efficiently decode. Moreover, to be sure that the Problem~2 always has
a solution, it is necessary and sufficient that $f_\code$ can decode
up to the covering radius of $\code$.  This means that solving
Problem~\ref{Prob:BSCP} with $\bound=\Rrho$ is precisely equivalent to
designing a stego-system which find solutions to both
Eqs.~(\ref{req:emb}) and (\ref{req:dist}) requirements for any
$\mathbf x$ and $\mathbf m$.  In this context, perfect
codes, for which the covering radius is precisely equal to the
error-correcting capacity ($\Rrho=\lfloor\frac{d-1}{2}\rfloor$), are
particularly relevant.

Unfortunately, using perfect codes leads to an embedding efficiency
which is far from the bound given in
Prop.~\ref{prop:embeffbound}~\cite{BierbrauerFridrich}. Hence
non-perfect codes have been studied (see the Introduction), even if
they can only tackle Problem~\ref{Prob:BSCP} for some $T$ much lower
than $\Rrho$. This may enable to force the system to perform only a
small number of modifications.

As discussed in the introduction, {\em Wet paper} codes were
introduced to improve embedding undetectability through the management
of locked, or \emph{wet}, components~\cite{Fridrich_writingon}.
\begin{problem}[Bounded syndrome wet paper coding problem]
  \label{Prob:BSCPLP}
  Let $\mathcal{C}$ be an $[n,k,d]_q$ linear code, $H$ be a parity
  check matrix of $\code$, $\mathbf{x} \in \mathbb{F}_q^n$, $\synd \in
  \mathbb{F}_q^{n-k}$, $\bound \in \mathbb{N}^*$, and a set of locked,
  or wet, components $\mathcal{I} \subset \lbrace1,\ldots,n\rbrace$,
  \mbox{$\ell = |\mathcal{I}|$}. The \emph{Bounded syndrome wet paper coding
    problem} consists in finding \mbox{$\mathbf{y}\in \mathbb{F}_q^n$} such
  that $\mathbf{y}H^t = \synd$, $d(\mathbf{x},\mathbf{y}) \le \bound$,
  and $\mathbf{x}_i = \mathbf{y}_i$ for all $ i \in \mathcal{I}$.
\end{problem}
Of course, solving Problem~\ref{Prob:BSCPLP} is harder and even
perfect codes may fail here.
More precisely, to deal with locked components, we usually decompose
the parity check matrix $H$ of $\code$ in the following
way~\cite{Fridrich_writingon,Fridrich06wetpaper}:
\begin{eqnarray*}
  \mathbf{y}H^t & = & \synd,\\
  \mathbf{y}_{|\bar{\mathcal{I}}}H_{|\bar{\mathcal{I}}}^t +
  \mathbf{y}_{|\mathcal{I}}H_{|\mathcal{I}}^t & = & \synd,\\
  \mathbf{y}_{|\bar{\mathcal{I}}}H_{|\bar{\mathcal{I}}}^t & = &
  \synd - \mathbf{y}_{|\mathcal{I}}H_{|\mathcal{I}}^t,
\end{eqnarray*}
where $\bar{\mathcal{I}}=\{1,\ldots,n\}\setminus \mathcal{I}$.
The previous equation can only be solved if
\mbox{$rank(H_{\bar{\mathcal{I}}}) = n-k$}. 
Since the potential structure of $H$ does not help to solve the
previous problem, we could as well choose $H$ to be also a random
matrix, which provides the main advantage to maximize asymptotically
the average embedding
efficiency~\cite{galand_itw_03,Fridrich06wetpaper}.

Hiding a long message  requires to split it and to
repeatedly use the basic scheme. Let $P_H$ the success probability for
embedding $(n-k)$ symbols, then the global success probability $P$ for
a long message of length $L(n-k)$ is $P_H^L$. This probability
decreases exponentially with the message length.

In order to bypass this issue, previous works propose either to take
another cover-medium, or to modify some locked components. In this
paper, we still keep unmodified the locked components, thus
maintaining the same level of undetectability. Moreover, we tackle the
particular case where the sender does not have a lot of cover-media
available, and needs a successful embedding, even if this leads to a
smaller embedding efficiency.

In the original  Wet Paper Setting
of~\cite{Fridrich_writingon}, the embedding efficiency is not dealt
with. In that case, we have a much easier problem.
\begin{problem}[Unbounded  wet paper Syndrome coding problem]
  \label{Prob:UBSCPLP}
  Let $\mathcal{C}$ be an $[n,k,d]_q$ linear code, $H$ be a parity
  check matrix of $\code$, $\mathbf{x} \in \mathbb{F}_q^n$, $\synd \in
  \mathbb{F}_q^{n-k}$, and a set of locked components $\mathcal{I}
  \subset \lbrace1,\ldots,n\rbrace$, $\ell = |\mathcal{I}|$. The
  \emph{Unbounded  wet paper Syndrome coding problem} consists in
  finding $\mathbf{y}\in \mathbb{F}_q^n$ such that $\mathbf{y}H^t =
  \synd$, and $\mathbf{x}_i = \mathbf{y}_i$, for all $ i \in
  \mathcal{I}$.
\end{problem}
In a random case setting, this problem can be discussed using a lower
bound on random matrices, provided by~\cite{Brent01randomkrylov}.
\begin{theorem}
  \label{Thm:proba}
  Let M be a random $ncol \times nrow$ matrix defined over $\mathbb{F}_q$,
  such that $ncol \ge nrow$. We have:
  $$
  	P\left(rank(M)=nrow\right) \ge
  	\left\lbrace
  	\begin{array}{ll}
  	  0.288, & \text{if }ncol=nrow \text{ and } q=2,\\
  	  1 - \frac{1}{q^{ncol-nrow}(q-1)}, & \text{otherwise}.
  	\end{array}
  	\right.
  $$ 
\end{theorem}
In a worst-case, or infallible, setting, the relevant parameter of
the code is its \emph{dual distance}.
\begin{proposition}\label{prop:BarbierMunuera}

  Consider a $q$-ary wet channel on length $n$ with at most
  $\ell$ wet positions, and that there exists a $q$-ary code $C$ whose
  dual code $C^\perp$ has parameters $[n,k^\perp,d^\perp=\ell]_q$ with
  $k^\perp+d^\perp=n+1-g$. Then we can surely embed $n-\ell-g$ symbols
  using a parity check matrix of $C$.
\end{proposition}
\begin{proof}
  This can be derived from~\cite[Theorem
  2.3]{BarbierMunuera:ARXIV2011}.\end{proof} This means that if the
code is $g$ far from the Singleton bound, then we loose $g$
information symbols with respect to the maximum. In particular, if
$n<q$, there exists a $q$-ary Reed-Solomon code with $g=0$, and we can
always embed $n-\ell$ symbols when there are $\ell$ wet
symbols. Coding theory bounds tells us that the higher $q$, the
smallest $g$ can be achieved, eventually using Algebraic-Geometry
codes~\cite{Vladut-Nogin-Tsfasman:AGCBN2007}.

\section{Randomized (wet paper) syndrome coding}
\label{sec:ourscheme}

Since embedding a message has a non-zero probability to fail, we
propose to relax the constraints in the following way:

\begin{problem}[Randomized bounded syndrome coding problem for wet paper]
  \label{PB:RWPC}
  Let $\mathcal{C}$ be an $[n,k,d]_q$ linear code, $H$ be a parity
  check matrix of $\code$, $r$ and $ \bound$ be two integers,
  $\mathbf{x} \in \mathbb{F}_q^n$, $\synd \in \mathbb{F}_q^{n-k-r}$ be
  the message to embed, and $\mathcal{I} \subset
  \lbrace1,\ldots,n\rbrace$ be the set of locked components,
  $\ell=|\mathcal{I}|$. Our \emph{randomized syndrome coding problem
    for wet paper} consists in finding $\mathbf{y}\in \mathbb{F}_q^n$
  and $\mathbf R\in\mathbb{F}_q^r$ such that (i) $\mathbf{y}H^t =
  (\synd||\mathbf{R})$, and $||$ denotes the concatenation operator,
  (ii) $d(\mathbf{x},\mathbf{y}) \le \bound$, and (iii) $\mathbf{x}_i
  = \mathbf{y}_i$, for all $ i \in \mathcal{I}$.
\end{problem}
We thus randomize one fraction of the syndrome to increase the number of
solutions. This gives a degree of freedom which may be large enough to
solve the system. The traditional approach can then be applied to find
$\mathbf{y}_{|\bar{\mathcal{I}}}$ and consequently $\mathbf{y}$.
Using some random symbols in the syndrome was used in the signature
scheme of Courtois, Finiasz and Sendrier
\cite{courtois-finiasz-sendrier-asiacrypt01}. While this reformulation
allows to solve the bounded syndrome coding problem in the wet paper
context without failure, we obviously lose some efficiency compared to
the traditional approach.

We now estimate the loss in embedding efficiency for a given
number of locked components. Let $e$ denote the embedding
efficiency of the traditional approach, and $e'$ denote the
efficiency of the randomized one. We obtain  a relative loss of:
$$
  \frac{e - e'}{e} = \frac{r}{n-k},
$$
while being assured that any $n-k-r$ message
be embedded, as long as $r<n-k$.

Optimizing the parameter $r$ is crucial, to ensure that our
reformulated problem always has a solution, while preserving the best
possible embedding efficiency. This is the goal on next Section.

\section{Case of perfect linear codes}
\label{sec:ex}

We discuss in this Section a sufficient condition on the size $r$ of
randomization, for our reformulated problem to always have a
solution.

\subsection{General Statement}

The \emph{syndrome function} associated with $H$, noted $S_H$, is
defined by:
  $$
  \begin{array}{rccl}
    S_H : & \mathbb{F}_q^n & \longrightarrow & \mathbb{F}_q^{n-k}\\
    & \mathbf{x} & \longmapsto & \mathbf{x}H^t.
  \end{array}
  $$
  This function $S_H$ is linear and surjective, and satisfies the
  following well-known properties. Let $\ball(\mathbf{x},T)$ denote
  the Hamming ball of radius $T$ centered on $\mathbf x$.

\begin{proposition}
  \label{Prop:Injective}
  Let $\mathcal{C}$ be an $[n,k,d]_q$-linear code, with covering
  radius $\Rrho$, $H$ a parity check matrix of $\mathcal{C}$, and
  $S_H$ the syndrome function associated with $H$. For all $\mathbf{x}
  \in \mathbb{F}_q^n$, the function $S_H$ restricted to
  $\ball(\mathbf{x},\left\lfloor\frac{d-1}{2} \right\rfloor)$ is
  one-to-one, the function $S_H$ restricted to
  $\mathcal{B}(\mathbf{x},\Rrho)$ is surjective.  When $\mathcal{C}$
  is perfect, the syndrome function restricted to
  $\ball(\mathbf{x},\Rrho)$ is bijective.
\end{proposition}
Now, we give a sufficient condition for upper-bounding $r$ in
Problem~\ref{PB:RWPC}.

\begin{proposition}
  \label{PROP:SC}
  Given a $[n,k,d]$ perfect code with $\Rrho\frac{d-1}2$, if the inequality
  \begin{eqnarray}
    \label{eqn:general}
    q^{n-k} + 1 \le q^r + \sum_{i=0}^{\Rrho}(q-1)^i {n-\ell\choose
    i},
  \end{eqnarray}
  is satisfied, then there exists a vector $\mathbf{y} \in
  \mathbb{F}_q^n$ and a random vector $\mathbf R$, which are solution
  of Problem~\ref{PB:RWPC}. In this case, Problem~\ref{PB:RWPC} always
  has a solution $\mathbf y$.
\end{proposition}
\begin{proof}
  Let $N_1$ ---respectively $N_2$--- be the number of different
  syndromes  generated by the subset of $\mathbb{F}_q^n$ satisfying
  (i) of Problem~\ref{PB:RWPC} --- respectively (ii) and (iii). If
  \begin{equation}
  	\label{EQ:SC}
  	N_1 + N_2 > q^{n-k}.
  \end{equation}
  Then there exists $\mathbf y$ which fulfills conditions (i), (ii),
  and (iii).  The number of different syndromes satisfying by the
  first constraint, for all $\mathbf R$, is $q^r$. Keeping in mind that
  $\ell$ components are locked and the syndrome function restricted to
  $\ball(\mathbf{x},\Rrho)$ is bijective, then
  $$
  N_2 = \sum_{i=0}^{\Rrho}(q-1)^i {n-\ell\choose i}.
  $$
  Combined with the sufficient condition~(\ref{EQ:SC}) we obtain the result.
\end{proof}
Next Section is devoted to the non trivial perfect codes: the Golay
codes, and the ($q$-ary) the Hamming codes.

\subsection{Golay codes}
\subsubsection{Binary Golay code}
We start by study the case of the binary $[23,12,7]_2$ Golay code, which
is perfect. The inequality of the proposition~\ref{PROP:SC} gives
\begin{equation}\label{G2}
r   \ge  \log_2\left( 1 + \frac{796}{3} \ell - \frac{23}{2} \ell^2  +
  \frac{1}{6} \ell^3\right).
\end{equation}

\subsubsection{Ternary perfect Golay code}
The ternary Golay code has parameters $[11,6,5]_3$. Using the
Proposition~\ref{PROP:SC}, we obtain:
\begin{equation}\label{G3}
  r  \ge  \log_3\left(1 + 44\ell - 2\ell^2 \right).
\end{equation}
Eqs~\ref{G2} and~\ref{G3} does not say much. We have plotted the
results in Fig.~\ref{Fig:SuccPb}, and we see that the number of
available bits for embedding degrades very fast with the number of
locked positions.

\begin{figure}[h]
  \centering \subfigure[Binary Golay code.]{
    \includegraphics[width=5.60cm]{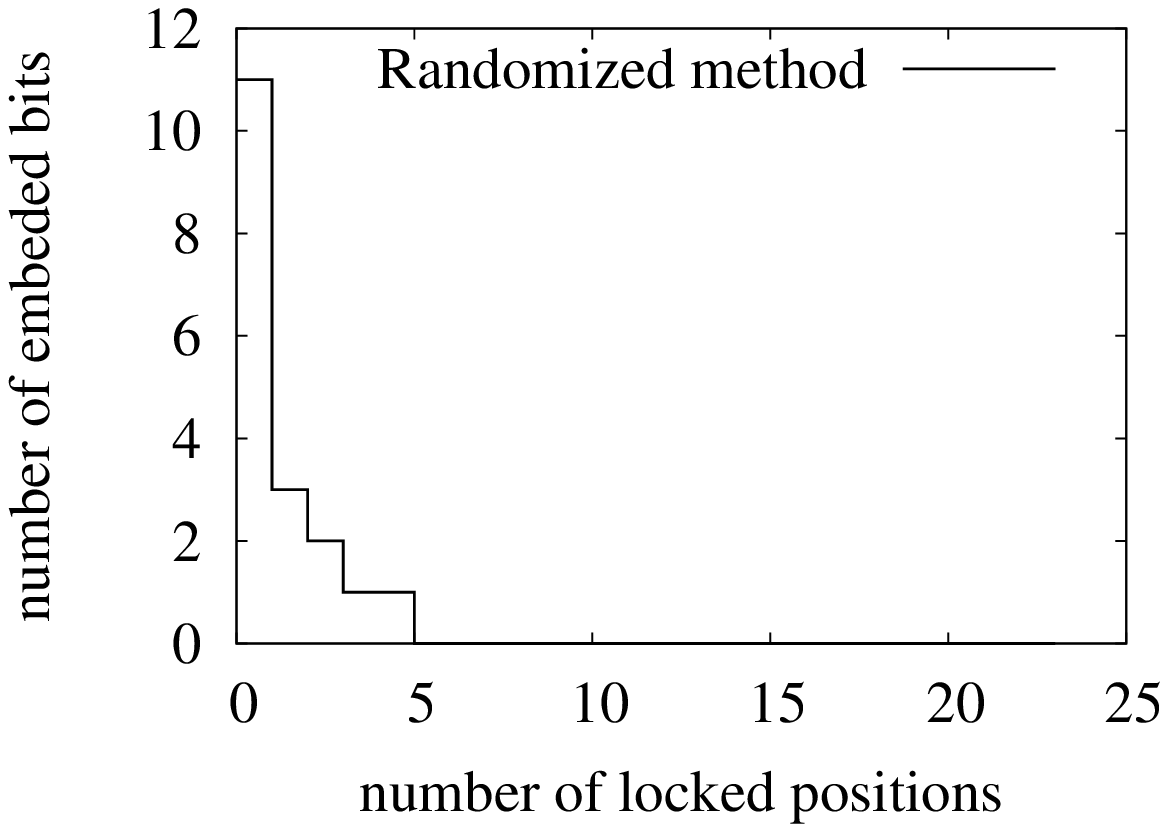}
  } \hfill \subfigure[Ternary Golay code.]{
    \includegraphics[width=5.60cm]{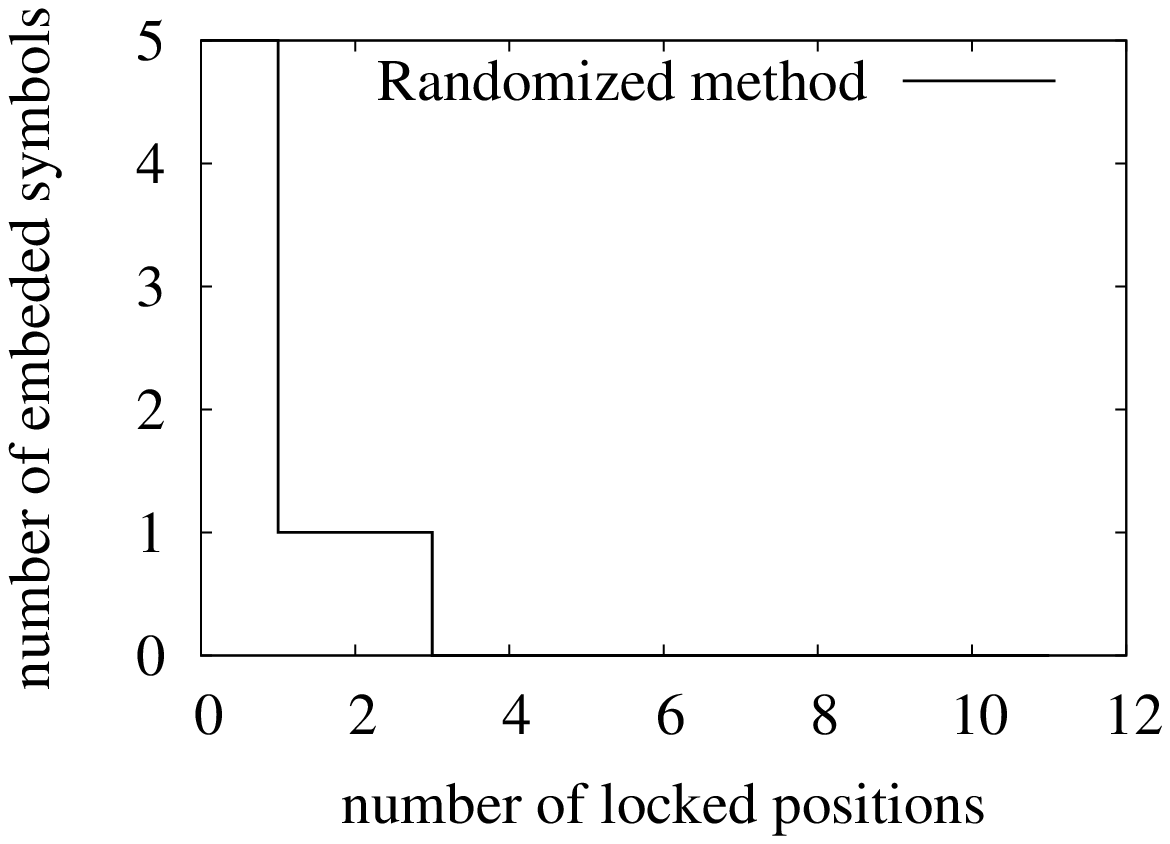}
  }
  \caption{Size of the random part for the two Golay codes. The number
    of remaining bits is plotted, in terms of the number of locked
    positions.}
  \label{Fig:SuccPb}
\end{figure}

\subsection{Hamming codes}
We study the infinite family of Hamming codes. We find $r$, analyze
the found parameters, and  study its asymptotic behavior.

\subsubsection{Computation of $r$}

Let $\mathcal{C}$ be a $[(q^\ext-1)/(q-1),n-\ext,3]_q$ Hamming code
over $\mathbb{F}_q$, for some $\ext$. Its covering radius is
$\Rrho=1$, and thus its embedding efficiency if $\ext$.  We aim to
minimize $r$, the length of the random vector $\mathbf{R}$. Since
$q^{n-k}=q^\ext$, $(q^\ext-1)/(q-1)=n$,
Proposition~\ref{PROP:SC} gives:
\begin{eqnarray}
r & \ge & \log_q \left( 1 + (q-1)\ell \right).\label{eqn:rHamming}
\end{eqnarray}

\subsubsection{Analysis of parameters}
In order to find an extreme case, it we maximize
the number of locked components $\ell$ while still keeping $n-k-r\geq
1$. A direct computation gives:
\begin{eqnarray*}
  \ext-1 & = & \log_q((q-1)\ell +1),\\
  \ell & = & \frac{q^{\ext-1}-1}{q-1} \approx \frac{n}{q}.\\
\end{eqnarray*}
Therefore, using Hamming codes, we can embed at least one information
symbol if no more than a fraction of $\frac1q$ of the components are
locked. This is of course best for $q=2$. The minimum $r$ which
satisfies inequality~(\ref{eqn:rHamming}) is $r= \lceil
\log_q((q-1)\ell +1) \rceil$. In other words, for Hamming codes, the
minimum number of randomized symbols needed to guarantee that the
whole message can be embedded, is logarithmic in the number of locked
components.  Our randomized approach always solves successfully
Problem~\ref{PB:RWPC} while traditional syndrome coding (including wet
paper) exhibits a non-zero failure rate, when $\frac\ell n< \frac 1q$.

\subsubsection{Asymptotic behavior}
Now we evaluate the loss in embedding efficiency.
Then, for a given~$\ell$,
the relative loss of the embedding efficiency is
given by:
$$
\frac{\lceil \log_q((q-1)\ell +1) \rceil}{\ext}.
$$
To conclude this section, we propose to focus on the normalized loss
in symbols for the family of Hamming codes. We assume that the rate of
$\ell$, the number of locked components to compare to $n$, the length
of the cover-data stays constant, i.e.\ $\ell = \r n$, for a given
$\r\in [0,\frac1q($. Then the asymptotic of relative  loss is
$$
  \frac{\log_q((q-1)\ell + 1)}{\ext} 
  \sim\frac{\log_q(n(q-1)\r)}{\ext}
\sim1+\frac{\log_q\r}\ext.
$$
This goes to $1$ when $p$ goes to infinity, i.e.\ all the symbols of
syndrome are consumed by the randomization. It makes sense, since
dealing with a given proportion $\lambda$ of arbitrarily locked
symbols in a long stego-data is much harder than dealing with several
smaller stego-data with the same proportion $\lambda$ of locked
positions.

\section{Using ZZW construction to embed dynamic parameters}
\label{sec:zzw}

In the approach given in previous Section, the sender and recipient
have to fix in advance the value of $r$. Indeed the recipient has to
know which part of syndrome is random. This is not very compliant with
the Wet Paper model, where the recipient does not know the quantity of
wet bits. We propose in this Section a variant of ZZW's
scheme~\cite{zhang_ih_08}, which enables to convey dynamically the
value $r$, depending on the cover-data.

\subsection{The scheme}
We consider that we are treating $n$ blocks of $2^p-1$ bits, $\mathbf
x_1,\dots,\mathbf x_{n}$, for instance displayed as in
Figure~\ref{fig:zzw}. Each block $\mathbf x_i$ is a binary vector of
length $2^p-1$, set as column, and we let $\mathbf v=(v_1,\dots,v_n)$
be the binary vector whose $i$-th coordinate $v_i$ is the parity bit
of column $\mathbf x_i$. We use the (virtual) vector $\mathbf v$ to
convey extra information, while at the same time  the $\mathbf x_i$ are
using for syndrome coding.

Our scheme is threefold : syndrome coding on the $\mathbf x_i$'s using
the parity check $H_1$ of a first Hamming code, with our randomized
method, then (unbounded wet paper) syndrome embedding on the syndromes
$\mathbf s_i$'s of the $\mathbf x_i$'s. This second syndrome embedding
see the $\mathbf s_i$ as $q$-ary symbols, and the matrix in use is the
parity check matrix $H_q$ of a $q$-ary Reed-Solomon code. We call the
$n$ first embeddings the $H_1$-embeddings, and the second one the
$H_q$-embedding. Finally, we use $\mathbf v$ to embed dynamic
information: the number $r$ of random bits, and $f$ the number of
failure in the $H_1$-embeddings. We call this last embedding the
$H_2$-embedding, where $H_2$ is the parity check matrix of a second,
much shorter, binary Hamming code.

We assume that $r$ is bounded by design, say $r\leq r_{\max}$.  We
shall see, after a discussion on all the parameters, that this is one
design parameter of the scheme, together with $o$, which the
precision, in bits, for describing real numbers $\in)\frac12,1]$.

\subsubsection{Embedding}\
 
\emph{ Inspect.} Each column $\mathbf x_1,\dots,\mathbf x_n$ is
inspected, to find the number of dry bits in each. This enables to
determine the size $r$ of the randomized part, which shall be the same
for all columns. This determines the columns $\mathbf x_i$'s where the
$H_1$-embeddings are feasible. Let $f$ be the number of $\mathbf
x_i$'s where the $H_1$-embeddings fail.

\emph{Build the wet channel.} For each of the $n-f$ columns $\mathbf
x_i$'s where the $H_1$-embedding is possible, there is a syndrome
$s_i$ of $p$ bits, where the last $r$ bits are random, thus wet, and
the $p-r$ first bits are dry.  We consider these blocks of $p-r$ bits as a
$q$-ary symbols, with $q=2^{p-r}$.  Thus we have a $q$-ary wet channel
with $n-f$ dry $q$-ary symbols, and $f$ wet $q$-ary symbols

\emph{Embed for the wet channel.} Then, using a Reed-Solomon over the
alphabet $\F_q$, we can embed $(n-f)$ $q$-ary symbols, using a
$n\times(n-f)$ $q$-ary parity check matrix $H_q$ of the code. Note
that the number of rows of this matrix is dynamic since $f$ is
dynamic.

\emph{Embed dynamic data.} We have to embed dynamic
  parameters $r$ and $f$ which are unknown to the recipient, using
  ZZW's virtual vector $\mathbf v$. For this binary channel, the dry
  bits $v_i$ correspond to the columns $\mathbf x_i$ where the
  $H_1$-embedding has failed, and where there is at least one dry
  bit in $\mathbf x_i$. A second Hamming code is used with parity
  check $H_2$ for this embedding.

\subsubsection{Recovery}\

\emph{$H_2$-extraction.} First compute $\mathbf v$, and using
the parity check matrix of the Hamming code $H_2$, extract $r$ and
$f$.  

\emph{$H_1$-extraction} Extract the syndromes of all the column
$\mathbf x_i$'s using the parity check matrix $H_1$, and collect only
the first $p-r$ bits in each column, to build $q$-ary symbols.

\emph{$H_q$-extraction} Build the parity check matrix $H_q$ of the
$q$-ary $[n,f]_q$ Reed-Solomon code, with $q=2^{p-r}$. Using this
matrix, get the $(n-f)$ $q$-ary information symbols, which are the
actual payload.
\begin{figure}
\begin{center}
\includegraphics[scale=.4]{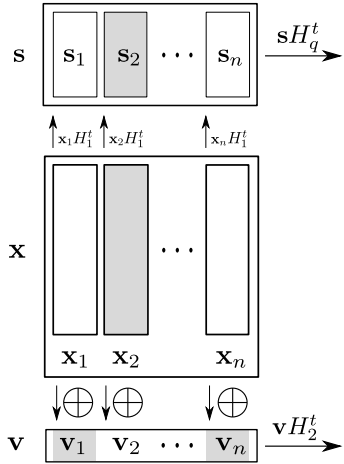}
\end{center}
\caption{A graphical view of our  scheme inspired from ZZW. A syndrome $\mathbf
  s_i$ is considered wet for the $H_q$-embedding when the
  $H_1$-embedding is not feasible. Then the corresponding bit $v_i$ in
  the vector $\mathbf v$ is dry for the $H_2$-embedding. Wet data is
  grey on the Figure.}
\label{fig:zzw}
\end{figure}
\subsection{Analysis}
There are several constraints on the scheme.

First, for a  Reed-Solomon code of length $n$ to exist over
the alphabet $\F_{2^{p-r}}$, we must have $n\leq 2^{p-r}$, for any
$r$, i.e.\ $n\leq 2^{p-r_{\max}}$. We fix \mbox{$n=2^{p-r_{\max}}-1$},
and let us briefly denote $u=p-r_{\max}$.

Then the binary $[n=2^{u}-1,2^u-u-1]_2$ Hamming code, with parity
check matrix $H_2$, is used for embedding in the vector $\mathbf v$,
with $f$ dry symbols. This a unbounded wet channel. From
Proposition~\ref{prop:BarbierMunuera}, we must have
\begin{align}\label{eq:min_f}
f&\geq 2^{u-1},
\end{align}
which implies that some columns $\mathbf x_i$ may be artificially
declared wet, for satisfying Eq.~\ref{eq:min_f}.
 Third, we also must have 
\begin{equation}\label{eq:u}
u=\ceil{\log r_{\max}}+\ceil{\log f_{\max}},
\end{equation}
to be able to embed $r$ and $f$. Since $f\leq 2^u-1$, we have
$\ceil{\log f_{\max}}=u$. Eq.~\ref{eq:u} becomes $u=\ceil{\log
  r_{\max}}+u$, this is clearly not feasible. To remedy this, instead
of embedding $f$, we embed its relative value $f_u=\frac f{2^u}\in[.5,
1] $, up to a fixed precision, say $o$ bits, with $o$ small. Then
Eq.~\ref{eq:u} is replaced by
\begin{align}\label{eq:u8}
  u&=\ceil{\log r_{\max}}+o,\\
  p&=r_{\max}+\ceil{\log r_{\max}}+o,
\end{align}
which is a condition easy to fulfill. It is also possible, by design,
to use the all-one value of $f_u$ as an out-of-range value to declare
an embedding failure. 
 The scheme is locally adaptive to the media:
for instance, in a given image, $r$ and $f$ may take different values
for different areas of the image.

In conclusion, the number of bits that we can embed using that scheme
is bounded by $ (n-f)(p-r)\leq 2^{u-1}(p-r), $ with dynamic  $r$ and $f$.

\section{Conclusion}
\label{sec:conclusion}
In this paper, we addressed the ``worst-case'' scenario, where the
sender cannot accept embedding to fail, and does not want relax the
management of locked components of his cover-data.  As traditional
(wet) syndrome coding may fail, and as the failure probability
increases exponentially with the message length, we proposed here a
different approach, which never fails.  Our solution is based on the
randomization of a part of the syndrome, the other part still carrying
symbols of  the message to transmit.  While our method suffers from a lost
of embedding efficiency, we showed that this loss remains acceptable
for perfect codes. Moreover, we showed how the size of the random part
of the syndrome, which is dynamically estimated during embedding, may
be transmitted to the recipient without any additional communication.

\bibliographystyle{splncs03}
\bibliography{biblio}

\end{document}